\documentclass[a4paper,USenglish,cleveref,thm-restate]{lipics-template/lipics-v2021}

\usepackage{hyperref}

\usepackage[algo2e, linesnumbered,ruled]{algorithm2e}

\usepackage{booktabs}
\usepackage[table]{xcolor}
\usepackage{siunitx}
\usepackage{tikz}
\usepackage{xcolor}
\usetikzlibrary{fit, calc, decorations.pathreplacing}

\usepackage{caption}
\usepackage{subcaption}

\usepackage{amsmath}
\usepackage{amsfonts}
\usepackage{amssymb}

\usepackage[capitalise]{cleveref}
\usepackage{etoolbox}
\usepackage{thmtools}

\usepackage{xspace}

\graphicspath{{./figures/}}

\hideLIPIcs

\title{Maintaining Discrete Probability Distributions in Practice}

\author{Daniel Allendorf}{Goethe-Universität Frankfurt}{daniel@ae.cs.uni-frankfurt.de}{https://orcid.org/0000-0002-0549-7576}{}

\authorrunning{D.~Allendorf}



\funding{This work was supported by the Deutsche Forschungsgemeinschaft (DFG) under grant ME 2088/5-1 (FOR~2975 --- Algorithms, Dynamics, and Information Flow in Networks).}

\supplement{Sources of implementations and experiments can be accessed at \url{https://github.com/Daniel-Allendorf/proposal-array}.}

\def\Oh#1{\ensuremath{\mathcal O\!\left(#1\right)}}
\def\impl#1{\normalfont{#1}}

\begin{document}

\maketitle

\normalsize

\begin{abstract}
	A classical problem in random number generation is the sampling of elements from a given discrete distribution.
	Formally, given a set of indices $S = \{1, \dots, n\}$ and sequence of weights $w_1, \dots, w_n \in \mathbb{R}^+$, the task is to provide samples from $S$ with distribution $p(i) = w_i / W$ where $W = \sum_j w_j$.
	A commonly accepted solution is Walker's Alias Table, which allows for each sample to be drawn in constant time.
	However, some applications correspond to a dynamic setting, where elements are inserted or removed, or weights change over time.
	Here, the Alias Table is not efficient, as it needs to be re-built whenever the underlying distribution changes.
	
	In this paper, we engineer a simple data structure for maintaining discrete probability distributions in the dynamic setting.
	Construction of the data structure is possible in time $\Oh{n}$, sampling is possible in expected time $\Oh{1}$, and an update of size $\Delta$ can be processed in time $\Oh{\Delta n / W}$.
	As a special case, we maintain an urn containing $W$ marbles of $n$ colors where with each update $\Oh{W / n}$ marbles can be added or removed in $\Oh{1}$ time per update.
	
	To evaluate the efficiency of the data structure in practice we conduct an empirical study.
	The results suggest that the dynamic sampling performance is competitive with the static Alias Table.
	Compared to existing more complex dynamic solutions we obtain a sampling speed-up of up to half an order of magnitude.
\end{abstract}

\setcounter{page}{1}

\section{Introduction}
\label{sec:introduction}
The task of sampling elements from a given discrete probability distribution occurs quite often as a basic building block for algorithms and software, for instance, in genetic algorithms \cite{golberg1989genetic, lipowski2012roulette}, simulation \cite{fox1995simulated, bratley2011guide, berenbrink2020simulating}, and generation of discrete structures \cite{hubschle2020linear, allendorf2023parallel}.
If only a small number of samples is required, then linear sampling suffices, i.e. for each sample we can in time $\Oh{n}$ draw a number $x \in [0, W]$ uniformly at random and identify the minimum index $i$ such that $x \leq \sum_{j=1}^i w_j$.
Often however, the number of required samples is much larger and in this case, it is preferable to first construct a data structure on top of the distribution which allows for the samples to be drawn more efficiently.
Moreover, some applications such as simulation are inherently dynamic, e.g. elements can be added or removed, or weights can change over time, and in this case, the task of the data structure becomes to maintain the probability distribution under such changes.
In this dynamic setting, a suitable data structure should allow for an efficient processing of updates to the underlying distribution, while also guaranteeing that the sampling performance does not degrade over time.

\subsection*{Related work}
A simple solution for the static case with optimal sampling time was given by Walker \cite{walker1977efficient}.
The method uses a data structure, called Alias Table, which contains in each column a pair of fractions of the weights of two indices such that the total weight of each pair sums to the mean $W / n$.
A sample can be drawn in time $\Oh{1}$ by selecting a column uniformly at random, and tossing a biased coin to select a row.
Kronmal and Peterson \cite{kronmal1979alias}, and Vose \cite{vose1991linear}, improved the construction time to $\Oh{n}$.
Bringmann and Larsen \cite{bringmann2013succinct} proposed more space-efficient variants.
Parallel construction algorithms for Alias Tables and their variants were given by H{\"u}bschle-Schneider and Sanders \cite{hubschle2022parallel}.

Maintaining a data structure which allows for efficient sampling even when the distribution is subject to changes over time is more challenging.
In particular, there is no known method to update the Alias Table without a full reconstruction even if only one weight changes by a small amount.
Moreover, while methods such as rejection sampling \cite{neumann1951various}, or storing the weights in a binary tree \cite{wong1980efficient}, are readily adapted to the dynamic setting \cite{d2022dynamic}, they generally do not allow for samples to be drawn in constant time.

A common approach taken by practicioners is to extend the Alias Table with a limited capability for updates.
For instance, \cite{berenbrink2020simulating} use a dynamized Alias Table as a building block for a simulation algorithm.
This method assumes an urn setting, where the weights are integers which indicate the number of marbles of a color.
Insertion or removal of a marble is possible in amortized constant time, but the total number of marbles must be at least the square of the number of colors.
Likewise, \cite{allendorf2023parallel} implement a graph generator by using an array which stores each index a number of times roughly proportional to its weight.
Samples are then drawn via rejection sampling and an increase in the weight of an index is processed by appending the index to the array an appropriate number of times, however, decreases in weight or removal of an index are not possible.

Taking on a more theoretical perspective, there also exist rather complex solutions which solve the problem in a near fully general setting or near optimal time. 
The data structure of Hagerup, Mehlhorn and Munro \cite{hagerup1993maintaining} provides samples in expected time $\Oh{1}$ and processes updates in time $\Oh{1}$ given that the weights $w_1, \dots, w_n$ are polynomial in $n$.
In complement, Matias, Vitter and Ni \cite{matias2003dynamic} proposes a data structure with expected sampling time $\Oh{\log^* n}$ and amortized expected update time $\Oh{\log^* n}$ with no restriction on the weights.
While these guarantees are suitable for most applications, implementing the required hierarchical data structures is rather difficult and their performance in practice is not known.

\subsection*{Our contribution}
We engineer a simple data structure for maintaining discrete probability distributions in the dynamic setting, called the Proposal Array.
Our method can be regarded as an extension and generalization of the methods used in \cite{berenbrink2020simulating} and \cite{allendorf2023parallel}.
It is intended as an easy to implement practical solution that complements the fully dynamic solutions of \cite{hagerup1993maintaining} and \cite{matias2003dynamic}.

We describe a static variant and two dynamic variants of the data structure.
The static Proposal Array can be constructed in time $\Oh{n}$ and allows for samples to be drawn in expected time $\Oh{1}$.
It uses rejection sampling similar to the Alias Table variant of \cite{bringmann2013succinct} but with different construction and rejection rules to facilitate dynamization.
Specifically, our rules are designed to avoid having to maintain a sorted order.
In comparison with the Alias Table the theoretical guarantee on the sampling performance is weaker but the construction algorithm is simpler and less susceptible to numerical errors due to floating point arithmetic.

The dynamic Proposal Array adds a procedure to update the data structure if the distribution changes.
Again each sample takes expected time $\Oh{1}$, and an update can be processed in time $\Oh{\Delta n / W}$ where $\Delta$ is the absolute difference in total weight between the previous and updated distribution.
The variant Proposal Array* improves the update procedure by eliminating the need to reconstruct the data structure after a certain number of updates.

The memory usage of all variants is dominated by maintaining an array which stores in each entry one of the indices $i \in S$.
For the static variant, the size of the array is at most $2n$ (see Lemma \autoref{th:prpll-con-fast}), and for the dynamic variants the size is at most $3n$ (see Lemmata \autoref{lem:p-suitable} and \autoref{lem:p-suitable-star}).

In comparison with the Dynamic Alias Table of \cite{berenbrink2020simulating}, we remove the restrictions on the weights, are able to process updates of larger size in constant time, and offer a variant which de-amortizes the processing time of updates.
Compared to \cite{allendorf2023parallel}, we add the possibility to decrease weights and remove indices, and analyze the resulting data structure in a general setting.

In an empirical study, we compare the sampling performance of the Proposal Array to the Alias Table and the data structures of \cite{hagerup1993maintaining} and \cite{matias2003dynamic}.
To the best of our knowledge, these are the first experiments to compare the performance of all the most well known dynamic solutions in practice.
The results suggest that dynamic sampling from the Proposal Array is similarly fast to static sampling from the Alias Table, and up to half an order of magnitude faster than sampling from the more complex solutions.
All our implementations including fast implementations of the Alias Table and \cite{hagerup1993maintaining} are openly available (see title page).

\subsection*{Organization}
Section \ref{sec:preliminaries} contains preliminaries and notation.
In Section \ref{sec:algorithm-1} we describe and analyze the static Proposal Array and in Sections \ref{sec:algorithm-2} and \ref{sec:algorithm-3} the dynamic variants.
Section \ref{sec:implementation} highlights optimization details of our implementations and Section \ref{sec:experiments} contains the experiments.

\section{Preliminaries}
\label{sec:preliminaries}
We denote a discrete probability distribution over $n$ elements as a pair $(S, \mathbf{w})$ which consists of an index set\footnote{It is possible to adapt our results to any other type of set by indexing the elements.} $S = \{1, \dots, n\}$ and a sequence of real-valued positive weights $\mathbf{w} = (w_1, \dots, w_n) \in \mathbb R_{>0}^n$.
The probability of index $i \in S$ under distribution $(S, \mathbf{w})$ is given by $p(i) = w_i / \sum_j w_j$.
For ease of notation, we denote the sum of all weights by $W = \sum_i w_i$ and use $\bar{w}$ as a shorthand for the mean $W / n$.

We analyze our algorithms and data structures in the randomized RAM \cite{DBLP:books/sp/Mehlhorn84} model of computation.
In particular, the following operations are assumed to take constant time: integer and fixed precision number arithmetic, the ceil and floor functions, computation of logarithms, and flipping a biased coin.

A \emph{growing array} is a data structure which can grow or shrink at the back and allows for random access to its memory locations.
For our purposes it suffices if each memory location stores one index $i \in S$.
We denote accessing location $l$ of an array $A$ by $A[l]$.
Locations are indexed $1, 2, \dots, |A|$.
Note that the maximum sizes of all arrays used by our algorithms can be determined up front if the maximum number of indices is known.
Thus, we can allocate all memory in advance which enables insertions at the back in time $\Oh{1}$ (as opposed to amortized time $\Oh{1}$).

For convenience, we define two multi-set like operations which are used to insert or delete an occurrence of an index without specifying the location.
\begin{itemize}
	\item Insert$(A, i)$: insert an occurrence of an index $i \in S$ into $A$.
	\item Erase$(A, i)$: erase an occurrence of an index $i \in S$ from $A$.
\end{itemize}

\noindent
Both operations can be implemented to take constant time with a standard trick.
\begin{lemma}
	\label{th:array-methods}
	Operations Insert$(A, i)$ and Erase$(A, i)$ can be implemented to take time $\Oh{1}$.
\end{lemma}

\begin{proof}
	For each index $i \in S$ we maintain an array $L_i$ containing all locations $l$ such that $A[l] = i$ and one global array $L$ containing for each $1 \leq l \leq |A|$ the entry $L[l] = k$ where $k$ is the location in $L_i$ such that $L_i[k] = l$.
	Now, to perform Insert$(A, i)$, we increase the size of $A$ by one, then write $i$ into the new empty location $|A|$, and finally, append $|A|$ at the back of $L_i$ and $|L_i|$ at the back of $L$.
	To perform Erase$(A, i)$, we remove the back elements $j$ from $A$, $k$ from $L$ and $l$ from $L_i$, and then set $A[l] \gets j$, $L_j[k] \gets l$ and $L[l] \gets k$.
\end{proof}

\section{Static Proposal Array}
\label{sec:algorithm-1}
We start by describing the method for the static case.
The idea is to maintain a growing array $P$, called the \emph{proposal array}, which contains each index $i \in S$ a number of times that is proportional to its weight $w_i$ up to a small rounding error.
We call the number of times that an index $i \in S$ is contained in $P$ the \emph{count} of $i$, and write this quantity as $c_i$.
It is useful to think of $P$ as a compression of the distribution $(S, \mathbf{w})$.
From this perspective, a loss of information arises due to the fact that $c_i$ has to take an integer value and cannot represent $w_i$ exactly.
Still, we can efficiently recover the distribution via rejection sampling: by selecting a uniform random entry in $P$, we propose index $i$ with probability proportional to $c_i$, and by accepting a proposed index $i$ as output with probability proportional to $w_i / c_i$, we output $i$ with probability proportional to $w_i$.

\subsection{Construction}

\begin{algorithm2e}[t]
	\KwData{Distribution $(S, \mathbf{w})$}
	\KwResult{Proposal array $P$ for $(S, \mathbf{w})$}
	Initialize array $P$\;
	Calculate mean $\bar{w} = W / n$\;
	\For{$i \in S$}{
		$c \gets \lceil w_i / \bar{w} \rceil$\;
		\For{$k \gets 1$ \KwTo $c $}{
			Insert $i$ into $P$\;
		}
	}
	Return $P$\;
	\caption{Construct}
	\label{algo:construct}
\end{algorithm2e}

To construct a suitable array $P$ for a given distribution $(S, \mathbf{w})$, we first calculate the mean $\bar{w} = W / n$.
Then, for each element $i \in S$, we compute\footnote{We stress that $w_i / \bar{w}$ scales with $n$. Thus, even for small weights $w_1, \dots, w_n < 1$ suitable counts are obtained.} $c = \lceil w_i / \bar{w} \rceil$, and insert $i$ into $P$ exactly $c $ times (see Algorithm \ref{algo:construct}).

\subsection{Sampling}

\begin{algorithm2e}[t]
	\KwData{Proposal array $P$ for $(S, \mathbf{w})$}
	\KwResult{Sample $i$ from $S$ with $p(i) = w_i / W$}
	\Repeat{\normalfont{Accepted}}{
		Select $i$ from $P$ uniformly at random\;
		Accept $i$ with probability $\frac{w_i}{c_i} / \max_{j \in S} \frac{w_j}{c_j}$\;
	}
	Return $i$\;
	\caption{Sample}
	\label{algo:sample}
\end{algorithm2e}

\noindent
Sampling is done in two steps (see Algorithm \ref{algo:sample}).
In step~1, we propose an index $i$ by selecting an entry of $P$ uniformly at random.
In step~2, we accept $i$ as output with probability $(w_i / c_i) / \max_{j \in S} (w_j / c_j)$, otherwise, we reject $i$, and restart from step 1.

\subsection{Analysis}
We now analyze the algorithms given in the previous subsections.
To start, we establish the correctness of the sampling algorithm.
\begin{theorem}
	\label{th:prpll-correct}
	Given a proposal array $P$ for the distribution $(S, \mathbf{w})$, Algorithm \ref{algo:sample} (Sample) outputs a given index $i \in S$ with probability $p(i) = w_i / W$.
\end{theorem}

\begin{proof}
	The probability of proposing and accepting a given index $i$ in any given iteration of the loop in steps $1-4$ of Algorithm \ref{algo:sample} is
	\begin{equation}
		\nonumber
		\underbrace{\frac{c_i}{|P|}}_{i \text{ is proposed}} \underbrace{\frac{w_i}{c_i \max_{j \in S} \frac{w_j}{c_j}}}_{i \text{ is accepted}} = \frac{w_i }{|P| \max_{j \in S} \frac{w_j}{c_j}} := \frac{w_i}{W'}.
	\end{equation}
	Regarding $W'$, it holds that
	\begin{equation}
		\nonumber
		W' = |P| \max_{j \in S} \frac{w_j}{c_j} \geq \sum_{1 \leq p \leq |P|} \frac{w_{P[p]}}{c_{P[p]}} = \sum_{j \in S} w_j = W
	\end{equation}
	and thus $W / W'$ is a probability.
	With the remaining probability $q = 1 - W / W'$, no index is accepted and the loop moves on to the next iteration.
	Thus, the overall probability of sampling $i$ is
	\begin{equation}
		\nonumber
		\frac{w_i}{W'} (1 + q + q^2 + \dots) = \frac{w_i}{W'} \sum_{k = 0}^\infty q^k = \frac{w_i}{W'} \frac{1}{1 - q} = \frac{w_i}{W}
	\end{equation}
	as claimed.
\end{proof}

\noindent
Next, we show the efficiency of the construction algorithm.
\begin{theorem}
	\label{th:prpll-con-fast}
	Given a distribution $(S, \mathbf{w})$ with $|S|~=~n$, Algorithm \ref{algo:construct} (Construct) outputs a proposal array $P$ for $(S, \mathbf{w})$ in time $\Oh{n}$.
\end{theorem}

\begin{proof}
	Observe that Construct runs in time linear in the size of the constructed array $P$.
	In addition, we have
	\begin{align}
		\nonumber
	 	|P| &= \sum_{i \in S} c_i = \sum_{i \in S} \left\lceil \frac{w_i n}{W} \right\rceil \leq \sum_{i \in S} \left(1 + \frac{w_i n}{W} \right)
		\\
		\nonumber
		&= n + \frac{\sum_{i \in S} w_i }{W} n = n + \frac{W}{W} n = 2 n
	\end{align}
	which shows the claim.
\end{proof}

\noindent
Finally, we show the efficiency of the sampling algorithm.
\begin{theorem}
	\label{th:prpll-sample-fast}
	Given a proposal array $P$ for a distribution $(S, \mathbf{w})$ output by Algorithm \ref{algo:construct} (Construct),
	Algorithm \ref{algo:sample} (Sample) runs in expected time $\Oh{1}$.
\end{theorem}

\begin{proof}
	Observe that the running time of Sample is asymptotic to the number of iterations of the loop in steps $1-4$ of Algorithm \ref{algo:sample}.
	The loop terminates if a proposed index is accepted, and the probability $p$ of accepting after any given iteration is
	\begin{align}
		\nonumber
		p = 1 - q = \frac{W}{W'} = \frac{W}{|P| \max_{j \in S} \frac{w_j}{c_j}}.
	\end{align}
	In the proof of Theorem \autoref{th:prpll-con-fast}, we have shown that
	\begin{align}
		\nonumber
		|P| \leq 2n
	\end{align}
	and by the construction rule $c_i = \lceil w_i / \bar{w} \rceil = \lceil w_i n / W \rceil$, it holds that
	\begin{align}
		\nonumber
		\max_{j \in S} \frac{w_j}{c_j} = \max_{j \in S} \frac{w_j}{\lceil w_j \frac{n}{W} \rceil} < \max_{j \in S} \frac{w_j}{w_j \frac{n}{W}} = \frac{W}{n}.
	\end{align}
	Thereby, we have
	\begin{align}
		\nonumber
		p > \frac{1}{2}
	\end{align}
	which implies that the expected number of iterations is less than $2$.
\end{proof}

\section{Dynamic Proposal Array}
\label{sec:algorithm-2}
We now dynamize the data structure given in the previous section.
To this end, we define a maintenance procedure Update which is called when the weight of a single index changes, or an index is added to or removed from the index set.

To model the state of the distribution and data structure at different time steps we consider the sequences $(S, \mathbf{w})^{(0)}, \dots, (S, \mathbf{w})^{(t)}$ and $P^{(0)}, \dots, P^{(t)}$.
In particular, $(S, \mathbf{w})^{(0)}$ denotes the initial distribution and $P^{(0)}$ the proposal array $P$ obtained by calling Algorithm \ref{algo:construct} (Construct) with the initial distribution as input.
Given some distribution $(S, \mathbf{w})^{(t)}$ with $t \geq 0$, if either (1) the weight of a single index changes, or (2) a new index is added, or (3) an index is removed, we write the new distribution as $(S, \mathbf{w})^{(t+1)}$.
Similarly, if $P^{(t)}$ is a proposal array for $(S, \mathbf{w})^{(t)}$, then we obtain $P^{(t+1)}$ by calling Algorithm \ref{algo:update} (Update) if (1) the weight of a single index changed, or if (2) a new index was added, or if (3) an index was removed.

Note that we will need to to call Algorithm \ref{algo:construct} (Construct) again if the mean $\bar{w}^{(t)} = W^{(t)} / n^{(t)}$ deviates too far from its initial value $\bar{w}^{(0)} = W^{(0)} / n^{(0)}$.
To this end, we keep track of the last step in which Construct was called via an additional variable $r$ with initial value $r \gets 0$.

\subsection{Update}
We now describe the Update procedure in detail.
Let $t$ denote the current step.
Then, if the weight of index $i$ is updated to a new value $w_i^{(t)}$, we update $P$ as shown in Algorithm~\ref{algo:update}.
First, we compute the new count as $c' = \lceil w_i^{(t)} / \bar{w}^{(r)} \rceil$.
Then, if $c'$ is larger than the old count $c_i$, we insert $i$ exactly $c' - c_i$ times into $P$.
Otherwise, if $c'$ is smaller than $c_i$, we erase $c_i - c'$ entries that contain $j$ from $P$.
Finally, if $\bar{w}^{(t)} < \bar{w}^{(r)} / 2$ or $\bar{w}^{(t)} > 2 \bar{w}^{(r)}$, we rebuild $P$ via Algorithm \ref{algo:construct} (Construct) and set $r \gets t$.

\begin{algorithm2e}[t]
	\KwData{Distribution $(S, \mathbf{w})^{(t)}$, index $i$ of changed weight, proposal array $P$, time of last reconstruction $r$}
	\uIf{$\bar{w}^{(t)} < \bar{w}^{(r)} / 2$ or $\bar{w}^{(t)} > 2 \bar{w}^{(r)}$}{
		$P \gets \text{Construct}((S, \mathbf{w})^{(t)})$\;
		$r \gets t$\;
	}
	\uElse{
		$c' \gets c_i$ \quad \quad \quad \quad \textbf{ if } $i \in S^{(t-1)}$ \textbf{ else } 0\;
		$c \gets \lceil w_i^{(t)} / \bar{w}^{(r)} \rceil$\hspace{0.11em} \textbf{ if } $i \in S^{(t)}$\quad \hspace{0.06em} \textbf{ else } 0\;
		\uIf{$c' > c$}{
			\For{$k \gets c$ \KwTo $c'$}{
				Insert $i$ into $P$\;
			}
		}
		\uElseIf{$c' < c$}{
			\For{$k \gets c'$ \KwTo $c$}{
				Erase $i$ from $P$\;
			}
		}
	}
	\caption{Update}
	\label{algo:update}
\end{algorithm2e}

\subsection{Analysis}
\label{subsec:analysis-2}
It is straightforward to check that Theorems \ref{th:prpll-correct} and \ref{th:prpll-con-fast} still apply, so we focus on the efficiency of the sampling and update procedures.

To help with the analysis we first give a condition under which a proposal array allows for efficient rejection sampling.
\begin{definition}
	\label{def:p-suitable}
	A proposal array $P$ is $\alpha$-\emph{suitable} for a distribution $(S, \mathbf{w})$ iff for all $i \in S$, we have
	\begin{equation}
		\nonumber
		\frac{1}{\alpha} \frac{w_i}{\bar{w}} \leq c_i \leq \left \lceil \alpha \frac{w_i}{\bar{w}} \right \rceil.
	\end{equation}
\end{definition}
As exemplified further below, the expected number of trials of rejection sampling for an $\alpha$-suitable proposal array is at most $\alpha^2 + \alpha$.
Thus, our goal is to show that $\alpha$ is a small constant for a proposal array maintained by the update procedure.

\begin{lemma}
	\label{lem:p-suitable}
	Algorithm \ref{algo:construct} (Construct) outputs a $1$-suitable proposal array $P^{(0)}$ for the initial distribution $(S, \mathbf{w})^{(0)}$.
	Furthermore, Algorithm \ref{algo:update} (Update) maintains a $2$-suitable proposal array $P^{(t)}$ for distribution $(S, \mathbf{w})^{(t)}$.
\end{lemma}

\begin{proof}
	The first claim follows immediately by the construction rule $c_i^{(0)} = \lceil w_i^{(0)} / \bar{w}^{(0)} \rceil$.
	For the second claim, observe that steps $1-3$ of Algorithm \ref{algo:update} (Update) guarantee that
	\begin{equation}
		\nonumber
		\frac{1}{2} \bar{w}^{(r)} \leq \bar{w}^{(t)} \leq 2 \bar{w}^{(r)}.
	\end{equation}
	Between reconstructions, the count $c_i^{(t)}$ of index $i$ always equals $c_i^{(t)} = \lceil w_i^{(t)} / \bar{w}^{(r)} \rceil$, so it holds that
	\begin{equation}
		\nonumber
		c_i^{(t)} = \left \lceil \frac{w_i^{(t)}}{\bar{w}^{(r)}} \right \rceil \geq \frac{w_i^{(t)}}{\bar{w}^{(r)}} \geq \frac{1}{2} \frac{w_i^{(t)}}{\bar{w}^{(t)}}
	\end{equation}
	and
	\begin{equation}
		\nonumber
		c_i^{(t)} = \left \lceil \frac{w_i^{(t)}}{\bar{w}^{(r)}} \right \rceil \leq \left \lceil 2 \frac{w_i^{(t)}}{\bar{w}^{(t)}} \right \rceil
	\end{equation}
	which completes the proof.
\end{proof}

\noindent
We can now show the efficiency of sampling from a dynamic Proposal Array.
\begin{theorem}
	\label{th:dyn-sample-efficient}
	Given an $2$-suitable proposal array $P$ for distribution $(S, \mathbf{w})$, Algorithm \ref{algo:sample} (Sample) runs in expected time $\Oh{1}$.
\end{theorem}

\begin{proof}
	Again the asymptotic running time of Sample equals the number of iterations of the loop in steps $1-4$ of Algorithm \ref{algo:sample} and the probability $p$ of exiting the loop after any given iteration is
	\begin{align}
		\nonumber
		p = 1 - q = \frac{W}{W'} = \frac{W}{|P| \max_{j \in S} \frac{w_j}{c_j}}.
	\end{align}
	As $P$ is $2$-suitable for the distribution, we have
	\begin{equation}
		\nonumber
		\frac{1}{2} \frac{w_i n}{W} = \frac{1}{2} \frac{w_i}{\bar{w}} \leq c_i \leq \left \lceil 2 \frac{w_i}{\bar{w}} \right \rceil = \left \lceil 2 \frac{w_i n}{W} \right \rceil
	\end{equation}
	for each $i \in S$.
	Therefore
	\begin{equation}
		\nonumber
		|P| = \sum_{i \in S} c_i \leq \sum_{i \in S} \left( 2 \frac{w_i n}{W} + 1 \right) = 3n
	\end{equation}
	and
	\begin{align}
		\nonumber
		\max_{j \in S} \frac{w_j}{c_j} \leq \max_{j \in S} \frac{w_j}{\frac{1}{2} w_j \frac{n}{W}} = 2 \frac{W}{n}
	\end{align}
	which implies
	\begin{equation}
		\nonumber
		p \geq \frac{1}{6}
	\end{equation}
	and thus the expected number of iterations is at most $6$.
\end{proof}

\noindent
Lastly, we show the efficiency of the update procedure.
\begin{theorem}
	\label{th:dynamic-update-fast}
	Given a proposal array $P$ for distribution $(S, \mathbf{w})$ and an updated distribution $(S, \mathbf{w})'$ with at most one index $i \in S \cup S'$ such that (1) $i \notin S$ and $w_i' = \Delta$ or (2) $i \notin S'$ and $w_i = \Delta$ or (3) $w_i' \neq w_i$ and $| w_i' - w_i | = \Delta$, Algorithm \ref{algo:update} (Update) runs in amortized time $\Oh{\Delta / \bar{w}}$.
\end{theorem}

\begin{proof}
	If the condition in step $1$ is met then the running time is dominated by the call to Algorithm \ref{algo:construct} (Construct) which takes time $\Oh{n}$ by Theorem \ref{th:prpll-con-fast}.
	The condition is met if the mean $\bar{w} = W / n$ halves or doubles with respect to its value at the most recent reconstruction, which requires at least $\Omega(W / \Delta)$ updates of size $\Delta$.
	Therefore, the amortized running time of a single update is $\Oh{n / (W / \Delta)} = \Oh{\Delta / \bar{w}}$.
	
	If the condition in step $1$ is not met, then the running time is asymptotic to $|c' - c|$, i.e. the number of entries of $P$ that need to be adjusted.
	Let $w_i = 0$ if $i \notin S$ and $w_i' = 0$ if $i \notin S'$.
	Then, in all three cases
	\begin{align}
		\nonumber
		|c' - c| &= \left | \left \lceil \frac{w_i'}{\bar{w}} \right \rceil - \left \lceil \frac{w_i}{\bar{w}} \right \rceil \right | 
		\\
		\nonumber
		&< 1 + \frac{| w_i' - w_i |}{\bar{w}} 
		\\
		\nonumber
		&= 1 + \frac{\Delta}{\bar{w}} = \Oh{\frac{\Delta}{\bar{w}}}
	\end{align}
	as claimed.
\end{proof}

\section{Dynamic Proposal Array*}
\label{sec:algorithm-3}

The method described in \cref{sec:algorithm-2} comes with the drawback that the array needs to be reconstructed after a certain number of updates.
This could be undesirable for applications where each update is under tight time constraints.
To address this issue, we describe a variant called Proposal Array* which de-amortizes the update procedure.
There are some standard techniques to de-amortize data structures such as \emph{lazy rebuilding} \cite{overmars1981worst}.
While the idea of the dedicated method we describe below bears similarities to lazy rebuilding, it comes with the advantage of being in-place, e.g. we only maintain one version of the data structure.
For simplicity, we limit the description and proofs to changes in weight.
An extension is straightforward by modifying the constants to account for the possibility that $n^{(t)} = n^{(t-1)} \pm 1$.

The idea is to augment the update procedure with additional steps which over time amount to the same result as a reconstruction.
In this way, the total amount of work remains unchanged but is now split fairly among individual updates according to update size.
Concretely, after every update we now check for indices with counts which would increase or decrease if the array was reconstructed.
To this end we use two additional variables $p, q$ initialized to $p \gets 1, q \gets n$ which can be imagined as pointers into an array containing the counts sorted by the indices.
Now, if the mean $\bar{w}$ increases, we check if the count of the index $j$ pointed to by $p$ must be decreased, and if so, erase an entry containing $j$ from $P$, otherwise, we move $p$ to the next index.
Similarly, if the mean decreases, we check if the count of the index $j$ pointed to by $q$ must be increased, and if so, insert an entry containing $j$ into $P$, otherwise, we move $q$ to the previous index.
Naturally, if a pointer hits a boundary, we reset it to its initial location.

What requires some attention is how many of the above maintenance steps are required.
As we want to guarantee a similar invariant as in Lemma \autoref{lem:p-suitable}, we must ensure that the counts of all indices have been checked and adjusted whenever the mean doubles or halves.
Thus, for a previous mean of $\bar{w}^{(t-1)}$ and updated mean of $\bar{w}^{(t)}$, we compute
\begin{equation}
	\nonumber
	s = \left \lceil 3n \log \left( \frac{\bar{w}^{(t)}}{\bar{w}^{(t-1)}} \right) \right \rceil
\end{equation}
and perform $|s|$ maintenance steps where $3n$ is an upper bound on the number of counts which may have to be checked (if $|s| > 3n$, we set $|s| = 3n$).

\subsection{Update*}

\begin{algorithm2e}[!t]
	\KwData{Distribution $(S, \mathbf{w})^{(t)}$, index $i$ of changed weight, proposal array $P$, maintenance pointers $p, q$}
	$c \gets c_i$\;
	$c' \gets \lceil w_i^{(t)} / \bar{w}^{(t)} \rceil$\;
	\uIf{$c' > c$}{
		\For{$k \gets c$ \KwTo $c'$}{
			Insert $i$ into $P$\;
		}
	}
	\uElseIf{$c' < c$}{
		\For{$k \gets c'$ \KwTo $c$}{
			Erase $i$ from $P$\;
		}
	}
	$s \gets \lceil 3n \log( \bar{w}^{(t)} / \bar{w}^{(t-1)}) \rceil$\;
	\For{$k \gets 1$ \KwTo $\min \{ |s|, 3n\}$}{
		$j \gets p \quad $\textbf{ if } $s > 0$ \textbf{ else } $q$\;
		$c \gets c_j$\;
		$c' \gets \lceil w_j^{(t)} / \bar{w}^{(t)} \rceil$\;
		\uIf{$c' > c$}{
			Insert $j$ into $P$\;
		}
		\uElseIf{$c' < c$}{
			Erase $j$ from $P$\;
		}
		\uElse{
			\uIf{$s > 0$}{
				$p \gets p + 1\quad $\textbf{ if } $p < n$ \textbf{ else } $1$\;
			}
			\uElse{
				$q \gets q - 1\quad $\textbf{ if } $q > 1$ \textbf{ else } $n$\;
			}
		}
	}
	\caption{Update*}
	\label{algo:update*}
\end{algorithm2e}

We now describe the modified Update routine, which we call Update*, in detail (see Algorithm \ref{algo:update*}).
After the weight of an element $i$ increases or decreases, we now update $P$ as follows.
We start by computing the new count of $i$ as $c' \gets \lceil w_i^{(t)} / \bar{w}^{(t)} \rceil$ and then insert or erase $i$ an appropriate number of times until its count is adjusted.
Next, we compute the number of maintenance steps as $s \gets \lceil 3n \log(\bar{w}^{(t)}/\bar{w}^{(t-1)}) \rceil$.
We then repeat the following steps $|s|$ times: (1) we determine the index $j$ pointed to by $p$ or $q$, (2) we compute an updated count of $j$ as $c'' = \lceil w_j^{(t)} / \bar{w}^{(t)} \rceil$, (3) if $c_j < c''$, insert $j$ into $P$, or if $c_j > c''$, erase one entry of $j$ from $P$, otherwise if $c_j = c''$, we move $p$ to the next index or $q$ to the previous index.

\subsection{Analysis}

We now show the equivalent of Lemma \autoref{lem:p-suitable} for Proposal Array*.
The $\Oh{1}$ expected sampling time then follows by Theorem \autoref{th:dyn-sample-efficient}.

\begin{lemma}
	\label{lem:p-suitable-star}
	Algorithm \ref{algo:update*} (Update*) maintains a $2$-suitable proposal array $P^{(t)}$ for distribution $(S, \mathbf{w})^{(t)}$.
\end{lemma}

\begin{proof}
	We show the upper bound, the proof of the lower bound is similar.
	Recall that the upper bound states that $c_i^{(t)} \leq \lceil 2 w_i^{(t)} / \bar{w}^{(t)} \rceil$ for all $i \in S$ in any step $t \geq 0$.
	The proof is by strong induction over $t$.
	
	For $t = 0$, the claim holds by the same argument as in Lemma \autoref{lem:p-suitable}.
	Now, assume that the claim holds in all steps $0, \dots, t - 1$, we will show that this implies the claim in step $t$.
	First, observe that the claim trivially holds if there is no step $0 \leq r < t$ such that $\bar{w}^{(r)} > \bar{w}^{(t)} / 2$.
	Otherwise, let $r$ be the most recent such step.
	Now, for each index $i \in S$, there are two possible cases.
	
	If the weight of index $i$ changed in some step $r < t' \leq t$ and $t'$ was the most recent such step, then by the adjustment rule for counts, it holds that $c_i^{(t)} = \lceil w_i^{(t')} / \bar{w}^{(t')} \rceil = \lceil w_i^{(t)} / \bar{w}^{(t')} \rceil \leq \lceil 2 w_i^{(t)} / \bar{w}^{(t)} \rceil$ where the inequality follows since assuming $\bar{w}^{(t')} > \bar{w}^{(t)} / 2$ would contradict the assumption that $r$ was the most recent step such that $\bar{w}^{(r)} > \bar{w}^{(t)} / 2$.
	
	Otherwise, we have $w_i^{(r)} = w_i^{(t)}$.
	In addition, the number of iterations of the loop in lines $12-24$ of Algorithm \ref{algo:update*} with $s > 0$ performed in steps $r + 1, \dots, t$ is at least
	\begin{align}
		\nonumber
		\sum_{k=r+1}^{t} \left \lceil 3n \log \left( \frac{\bar{w}^{(k)}}{\bar{w}^{(k-1)}}\right) \right \rceil &\geq \sum_{k=r+1}^{t} 3n \log \left( \frac{\bar{w}^{(k)}}{\bar{w}^{(k-1)}}\right)
		\\
		\nonumber
		&= 3n \log \left( \prod_{k=r+1}^{t} \frac{\bar{w}^{(k)}}{\bar{w}^{(k-1)}} \right)
		\\
		\nonumber
		&= 3n \log \left( \frac{\bar{w}^{(t)}}{\bar{w}^{(r)}}\right) \geq 3n.
	\end{align}
	Also, the upper bound implies $|P|^{(t)} = \sum_i c_i^{(t)} \leq 3n$, and since the upper bound holds in steps $0, \dots, t-1$ by the induction hypothesis, we have performed at least one iteration of the loop for each count of each index in steps $r + 1, \dots, t$.
	In particular, we had $p = i$ in some step $r < t' \leq t$, and since in each iteration where $p = i$, we decrease the count of $i$ unless $c_i^{(t)} \leq \lceil w_i^{(t')} / \bar{w}^{(t')} \rceil$ we have $c_i^{(t)} \leq \lceil w_i^{(t')} / \bar{w}^{(t')} \rceil = \lceil w_i^{(t)} / \bar{w}^{(t')} \rceil \leq \lceil 2 w_i^{(t)} / \bar{w}^{(t)} \rceil$.
	
	Thus the claim holds in step $t$ in particular and by strong induction in all steps $t \geq 0$.
\end{proof}

\noindent
It only remains to analyze the running time of Algorithm \ref{algo:update*} (Update*).

\begin{theorem}
	\label{th:dynamic-update*-fast}
	Given a proposal array $P$ for distribution $(S, \mathbf{w})$ and an updated distribution $(S, \mathbf{w})'$ with at most one index $i \in S$ such that $w_i' \neq w_i$, Algorithm \ref{algo:update*} (Update*) runs in time $\Oh{\Delta / \bar{w}}$ where $\Delta = | w_i' - w_i |$.
\end{theorem}

\begin{proof}
	The running time bound on steps $1-10$ follows by the same argument as for steps $5-14$ of Algorithm \ref{algo:update} in the proof of Theorem \ref{th:dynamic-update-fast}.
	The running time of the remaining steps is asymptotic to
	\begin{align}
		\nonumber
		|s| &= \left | \left \lceil 3n \log \left( \frac{\bar{w}'}{\bar{w}} \right) \right \rceil \right | 
		\\
		\nonumber
		&\leq \left | \left \lceil 3n \log \left( \frac{W'}{W} \right) \right \rceil  \right |
		\\
		\nonumber
		&\leq \left | \left \lceil 3n \log \left( \frac{W + w_i' - w_i}{W} \right) \right \rceil  \right |
		\\
		\nonumber
		&\leq \left | \left \lceil 3 n \frac{w_i' - w_i}{W}\right \rceil \right | = \Oh{\frac{\Delta}{\bar{w}}}
	\end{align}
	where the last inequality follows by $\log (1 + x) \leq x$ for $x > -1$.
\end{proof}

\section{Implementation Details}
\label{sec:implementation}
This section highlights some details of our implementations.
Our main focus is reducing unstructured memory accesses which easily dominate the sampling time of the data structures.

For our implementation of the Alias Table (see \autoref{subsec:experiments-static}), we store the table as an array of structs rather than multiple arrays which allows us to look up an index, alias, and threshold with one access to memory.

For our Proposal Array implementations, we reduce the number of memory accesses to one per rejection sampling attempt.
The idea is to regard each entry of the array $P$ as a bucket of capacity $\bar{w}$.
We then split up the total weight $w$ of an index with count $c$ among $c - 1$ full buckets and a final bucket to store the remainder $r = w - (c - 1) \bar{w}$.
As an entry associated with a full bucket can be accepted with probability $1$, it suffices to store the remainders in the first $n$ positions of $P$ and the indices of the full buckets in the remaining $|P| - n$ positions.
When looking up a random location $l$ of $P$, we now either have $l \leq n$, and accept index $i = l$ with probability $P[l] / \bar{w}$, or we have $l > n$, in which case we immediately accept and return the index $i = P[l - n + 1]$.
The same approach is straightforward to translate to the Dynamic Proposal Array (\autoref{sec:algorithm-2}) as the mean $\bar{w}^{(r)}$ is fixed between reconstructions.

Optimizing the implementation of the Dynamic Proposal Array* (\autoref{sec:algorithm-3}) is more challenging as the mean $\bar{w}^{(t)}$ is not fixed which prevents us from assuming the same capacity for full buckets of different indices.
For this reason we first restrict the permissible capacities of the buckets to powers of two times the initial mean $\bar{w}^{(0)}$.
The maintenance routine then guarantees that all full buckets at a given step $t$ have one of two possible capacities.
This allows us to infer the necessary correction to the acceptance probability of a proposed index on the fly by comparing the index to the position of a common maintenance pointer which assumes the roles of both $p$ and $q$.
Moreover, we are able to perform the correction of the acceptance probability for full buckets in a rejection-free way by exploiting that the capacities of buckets are related by powers of two.

\section{Experiments}
\label{sec:experiments}
In this section, we study the previously discussed algorithms and data structures empirically.
Sources of our implementations, benchmarks, and tests are openly available (see title page).

The benchmarks are built with \texttt{GNU g++-11.2} and \texttt{rustc-1.68.0-nightly} and executed on a machine equipped with an Intel Core i5-1038NG processor and 16~GB RAM running macOS 12.6.
Pseudo-random bits are generated using the MT19937-64 variant of the Mersenne Twister~\cite{DBLP:journals/tomacs/MatsumotoN98}.

\subsection{Static Data Structures}

\label{subsec:experiments-static}
In Section \ref{sec:introduction} we argue that the static Proposal Array is easier to construct and dynamize than the Alias Table.
However, it could be the case that these benefits are outweighed by the less efficient rejection based sampling procedure.
Therefore, we first compare the performance of the static Proposal Array to our implementation of the Alias Table\footnote{We also considered \texttt{std::discrete\_distribution} from the c++ standard library. However, during preliminary experiments, we found that this implementation was $7-12$ times slower.}.
As a reference point to a more naive sampling method, we additionally implement a binary tree which allows for sampling in time $\Oh{\log n}$.
The binary tree implementation is reasonably optimized, in particular, we only call the RNG once per sample.

\subsubsection*{Construction}

\begin{figure}[!t]
  \centering
  \includegraphics[scale=0.55]{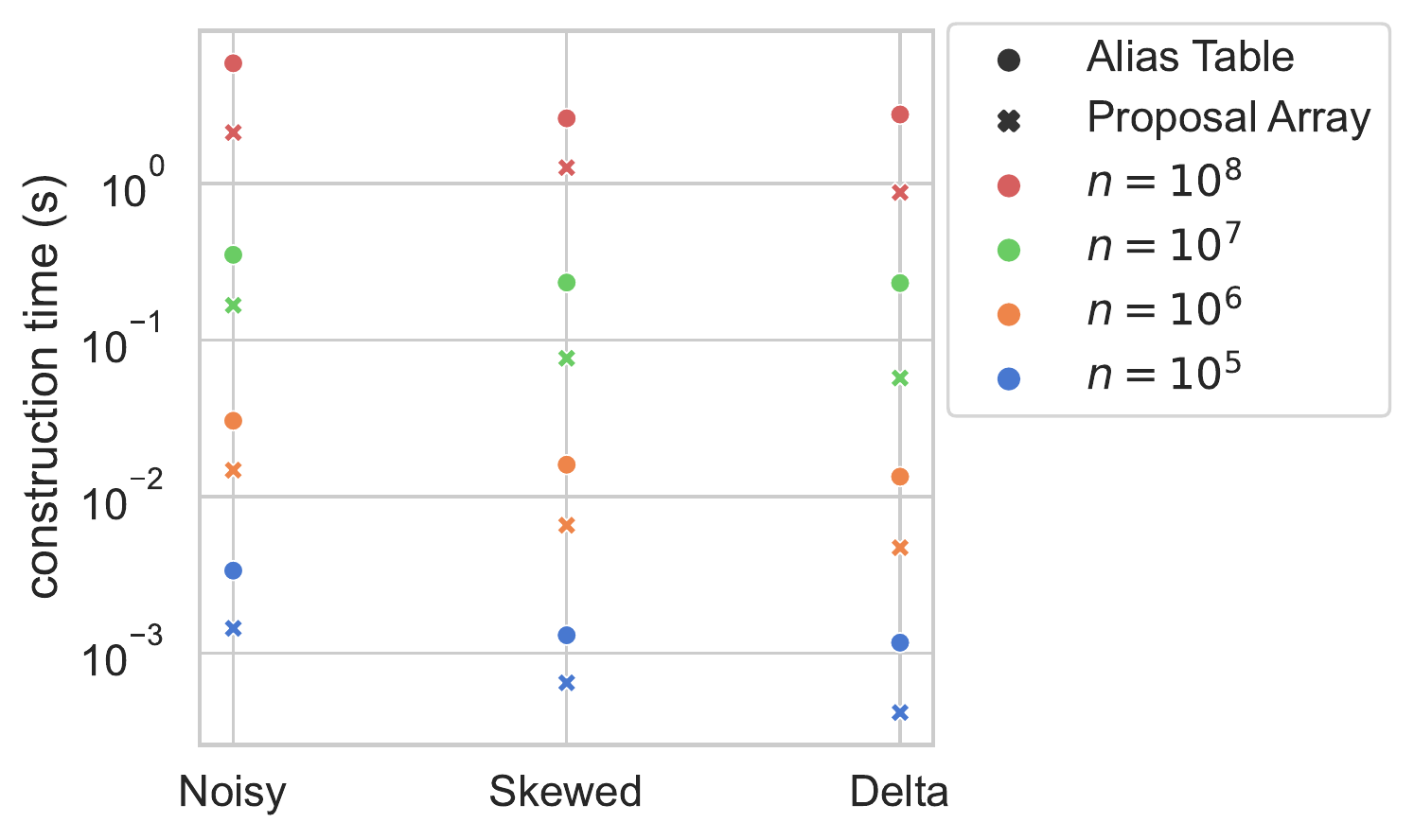}
  \caption{Average construction times of Alias Table and Proposal Array in seconds for index sets of various sizes.}
 \label{fig:static-construction}
\end{figure}

We start by comparing the performance of the construction procedures of the Alias Table and Proposal Array.
As input we use the following types of weight distributions on an index set $S = \{1, \dots, n\}$.
\begin{itemize}
	\item \textbf{Noisy:} For each index $i \in S$ we draw a real-valued weight $w_i \in \mathbb R$ uniformly at random from the interval $[0, n)$.
	\item \textbf{Skewed:} For each index $i \in S$ we draw an integer weight $w_i \in \mathbb N_{>0}$ with $p(w_i = k) \propto 1 / k^2$, i.e. the weights follow a power-law distribution.
	\item \textbf{Delta:} We draw $n - 1$ real-valued weights uniformly at random from $[0, 1)$ and set the remaining weight to $w_n = n$.
\end{itemize}

\noindent
\autoref{fig:static-construction} shows the results.
On all inputs, the construction of the Proposal Array is faster than the construction of the Alias Table by a factor of at least $2$.

\subsubsection*{Sampling}

\begin{figure}[!t]
  \centering
  \includegraphics[scale=0.55]{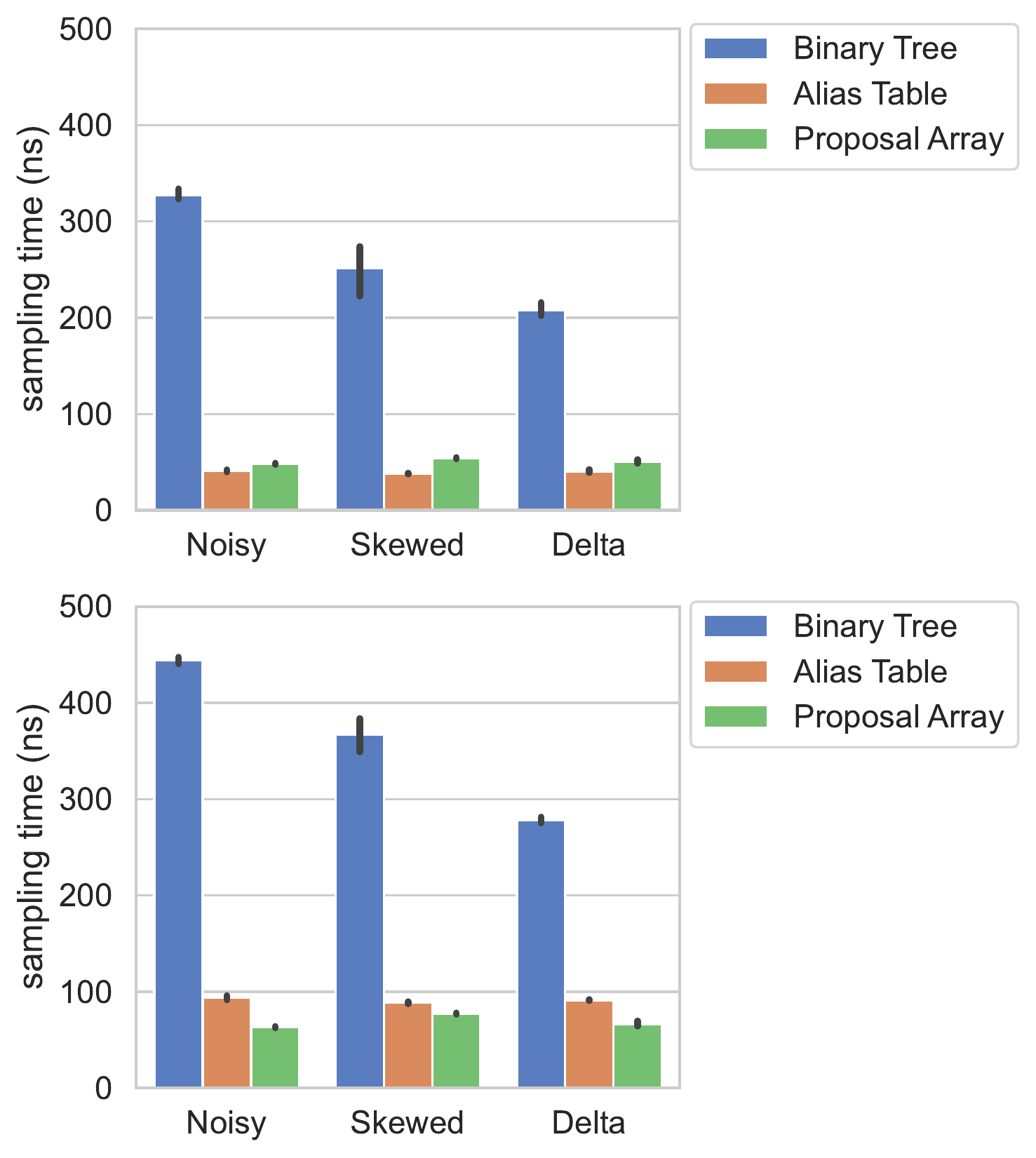}
  \caption{Average sampling times of Alias Table, Proposal Array and Binary Tree in nanoseconds on index sets of size $n=10^7$ (top) and $n = 10^8$ (bottom). The black vertical bars indicate the $95\%$ confidence interval.}
  \label{fig:static-sampling}
\end{figure}

In the second experiment we evaluate the sampling time of the data structures.
As input we consider the same three types of weight distributions as above for index sets of size $n = 10^7$ and $n = 10^8$.
We then measure the time required to draw $10^6$ samples and plot the average time per sample.

\autoref{fig:static-sampling} summarizes the results.
Sampling from a Proposal Array is slower than sampling from an Alias Table for $n = 10^7$ but faster for $n = 10^8$.
The binary tree is not competitive with the constant sampling time data structures.

Overall the results suggest that it is reasonable to pursue the rejection sampling approach for the benefit of easier dynamization.

\subsection{Dynamic Data Structures}

We move on to evaluating the dynamic \impl{Proposal Array} and the variant \impl{Proposal Array*}.

\subsubsection*{Processing of Updates}

\begin{figure}[!t]
 \centering
 \begin{subfigure}{\columnwidth}
  \centering
  \includegraphics[scale=0.55]{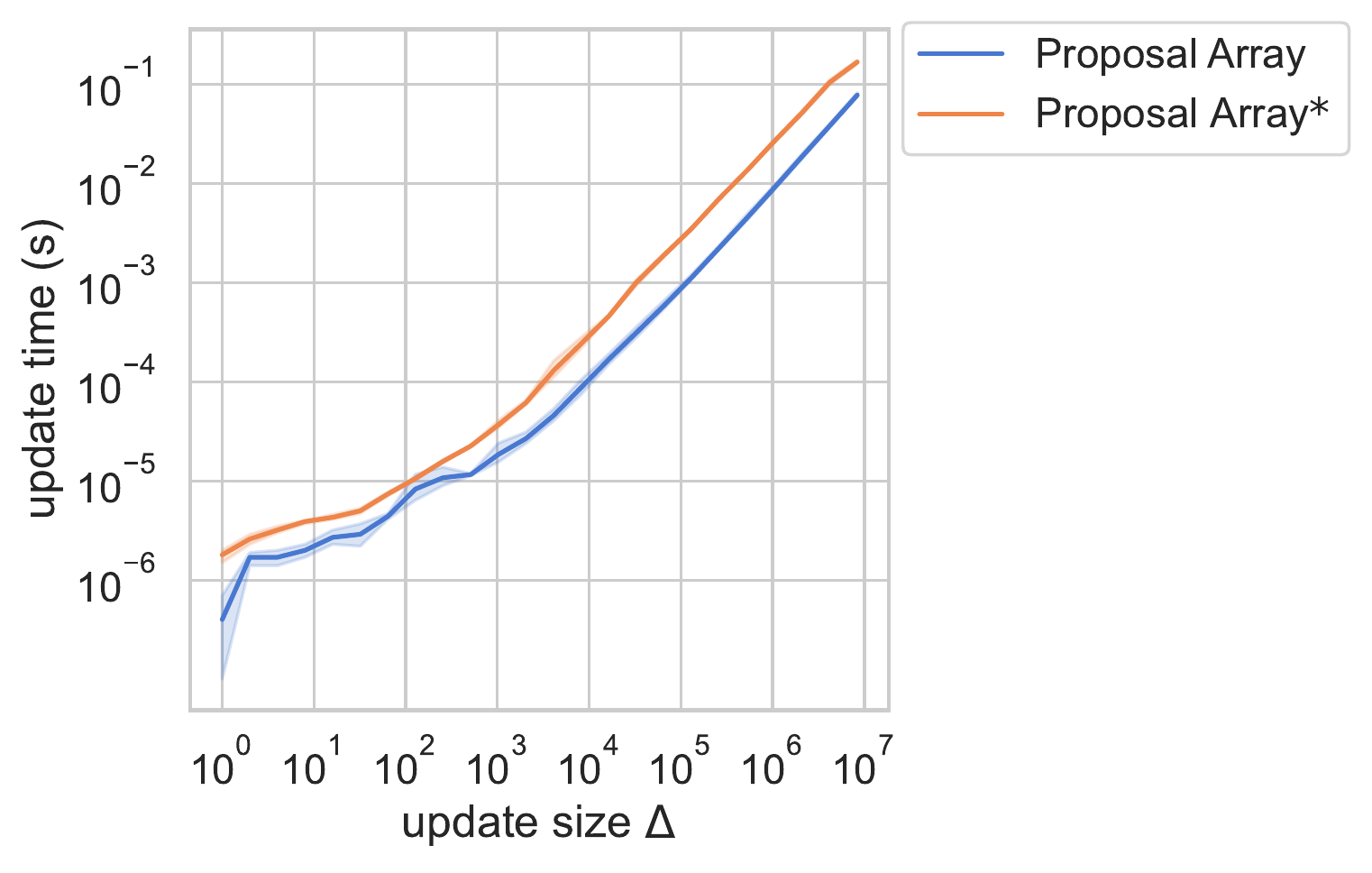}
  \caption{Processing times of updates of exponentially increasing size.}
 \label{fig:dynamic-increasing}
 \end{subfigure}
 \vfill
 \begin{subfigure}{\columnwidth}
  \centering
  \includegraphics[scale=0.55]{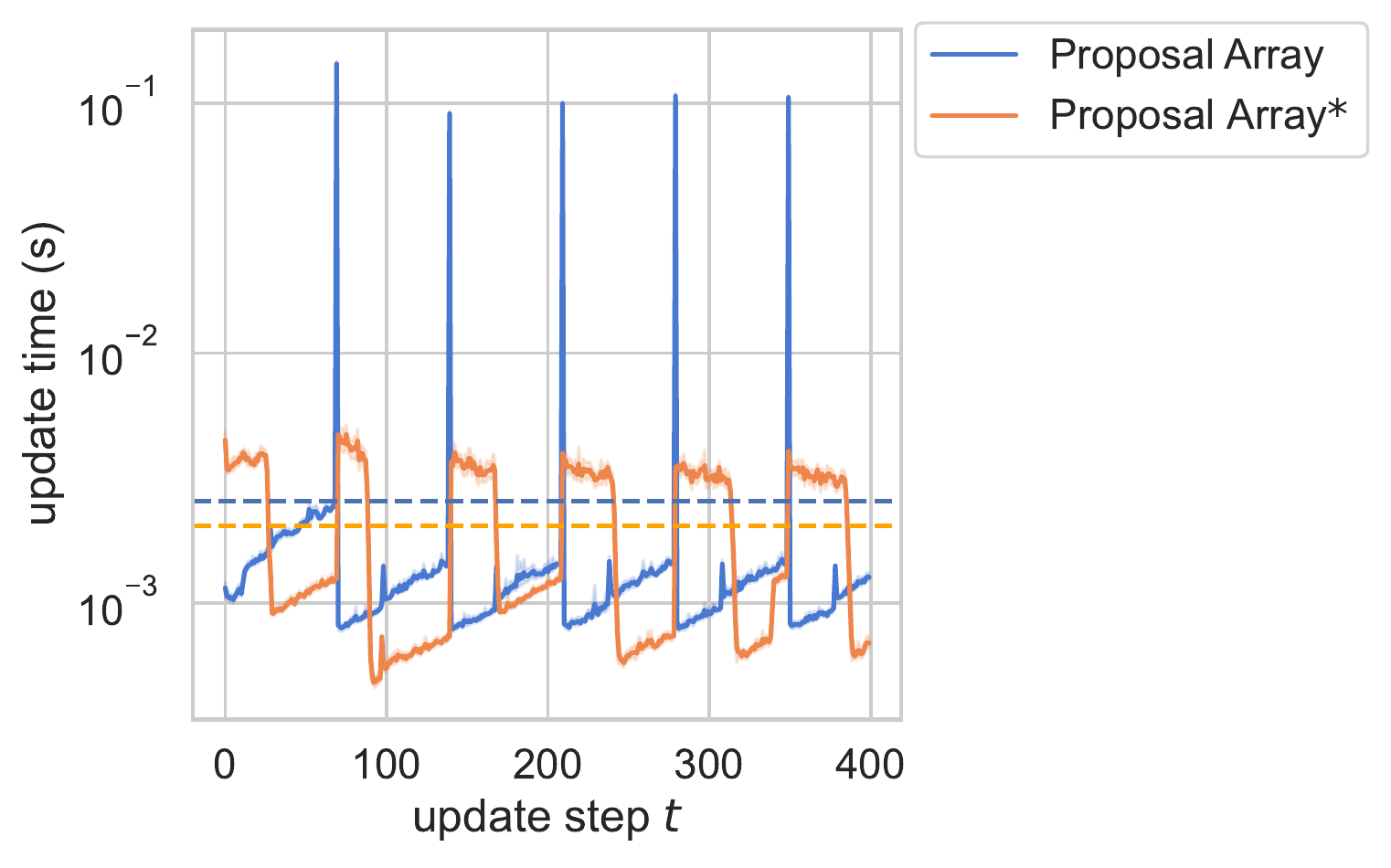}
  \caption{Processing times for updates of size $\Delta = 10^6 \bar{w}^{(t)}$. The dashed lines indicate the means over all steps.}
  \label{fig:dynamic-constant}
 \end{subfigure}
 \caption{Running times of the update procedures of \impl{Proposal Array} and \impl{Proposal Array*}. The shaded regions indicate the $95\%$ confidence interval.}
\end{figure}

We start by comparing the performance of the update procedures of \impl{Proposal Array} and \impl{Proposal Array*} to evaluate the benefit of avoiding reconstructions.

Starting from a uniform distribution $w_i = 1$ for $i \in S$ with $n = 10^7$ we increase the weight of random indices and measure the average processing times.
For updates of increasing size we observe linear scaling in the update size for both data structures (see \autoref{fig:dynamic-increasing}).
In contrast, scaling the update size with $\bar{w}^{(t)}$ causes the average processing time to remain constant (see \autoref{fig:dynamic-constant}).
While the processing time oscillates for both data structures, spikes for the regular \impl{Proposal Array} are up to two orders of magnitude larger due to the high cost of reconstructions. 
In most steps updates of \impl{Proposal Array*} are slower due to the additional work required.
However, on average updates of \impl{Proposal Array*} are faster by a factor of $1.3$, which suggests that the overhead of a full reconstruction exceeds the additional work required for the gradual reconstruction.

\subsubsection*{Dynamic Sampling}

\begin{figure}[!t]
 \centering
 \begin{subfigure}{\columnwidth}
  \centering
 \includegraphics[scale=0.5]{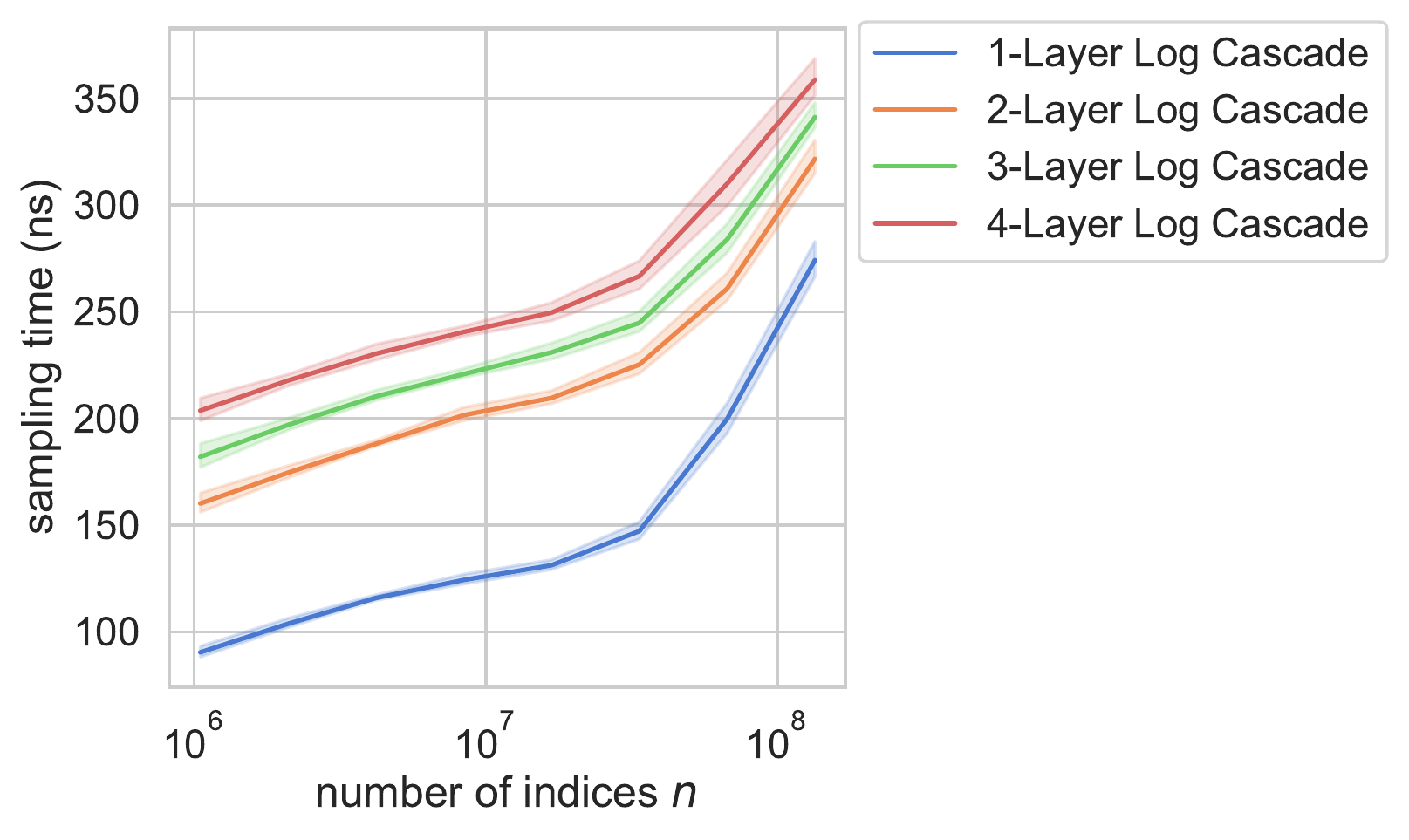}
 \end{subfigure}
 \caption{Sampling performance of \impl{Log Cascade} implementations with varying numbers of layers. Each data point is averaged over $10$ random weight distributions of the \textbf{Noisy} type (see \autoref{subsec:experiments-static}). The shaded regions indicate the $95\%$ confidence interval.}
 \label{fig:log-cascade-layers}
\end{figure}

We move on to evaluating the dynamic sampling performance.

For a comparison with the state of the art we use the data structures of \cite{hagerup1993maintaining} and \cite{matias2003dynamic}.
Our implementation of \cite{hagerup1993maintaining} uses one partition layer rather than multiple layers with a table look-up, which we find to result in a superior performance in practice (see \autoref{fig:log-cascade-layers} and \autoref{sec:a-log-cascade}).
As implementation of \cite{matias2003dynamic} we use \texttt{dynamic-weighted-index}\footnote{\url{https://crates.io/crates/dynamic-weighted-index}} (see \cite{allendorf2023parallel} for implementation details).
While this implementation is only available in the rust programming language, we believe that the languages are similar enough in performance to allow for a meaningful comparison.
In the following, we refer to these implementations as \impl{Log Cascade} and \impl{Weighted Index}, respectively.
Finally, we include a dynamized variant of the \impl{Binary Tree} used in \autoref{subsec:experiments-static}.

\begin{figure}[t!]
 \centering
 \begin{subfigure}{\columnwidth}
  \centering
  \includegraphics[scale=0.55]{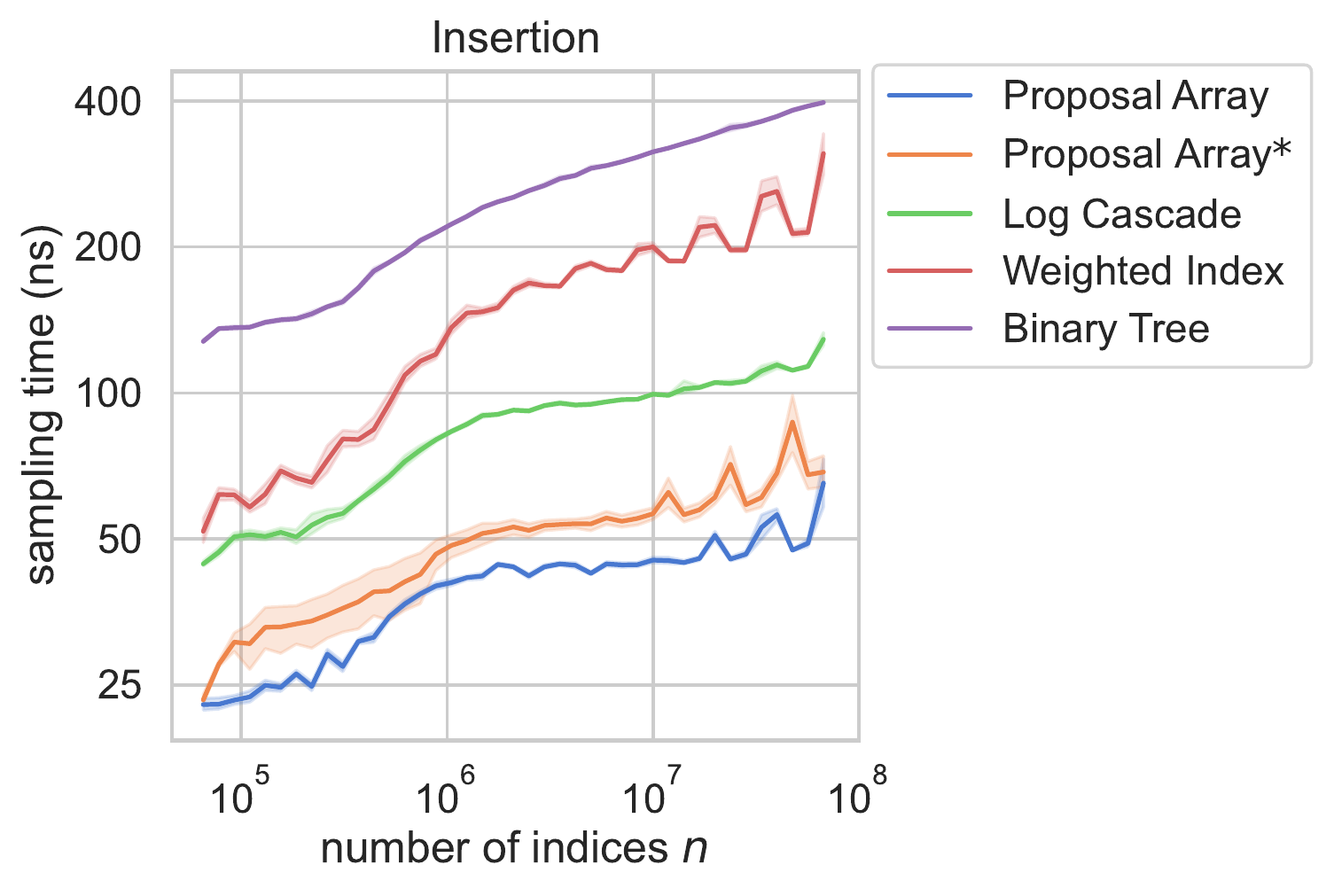}
 \end{subfigure}
 \vfill
 \begin{subfigure}{\columnwidth}
  \centering
  \includegraphics[scale=0.55]{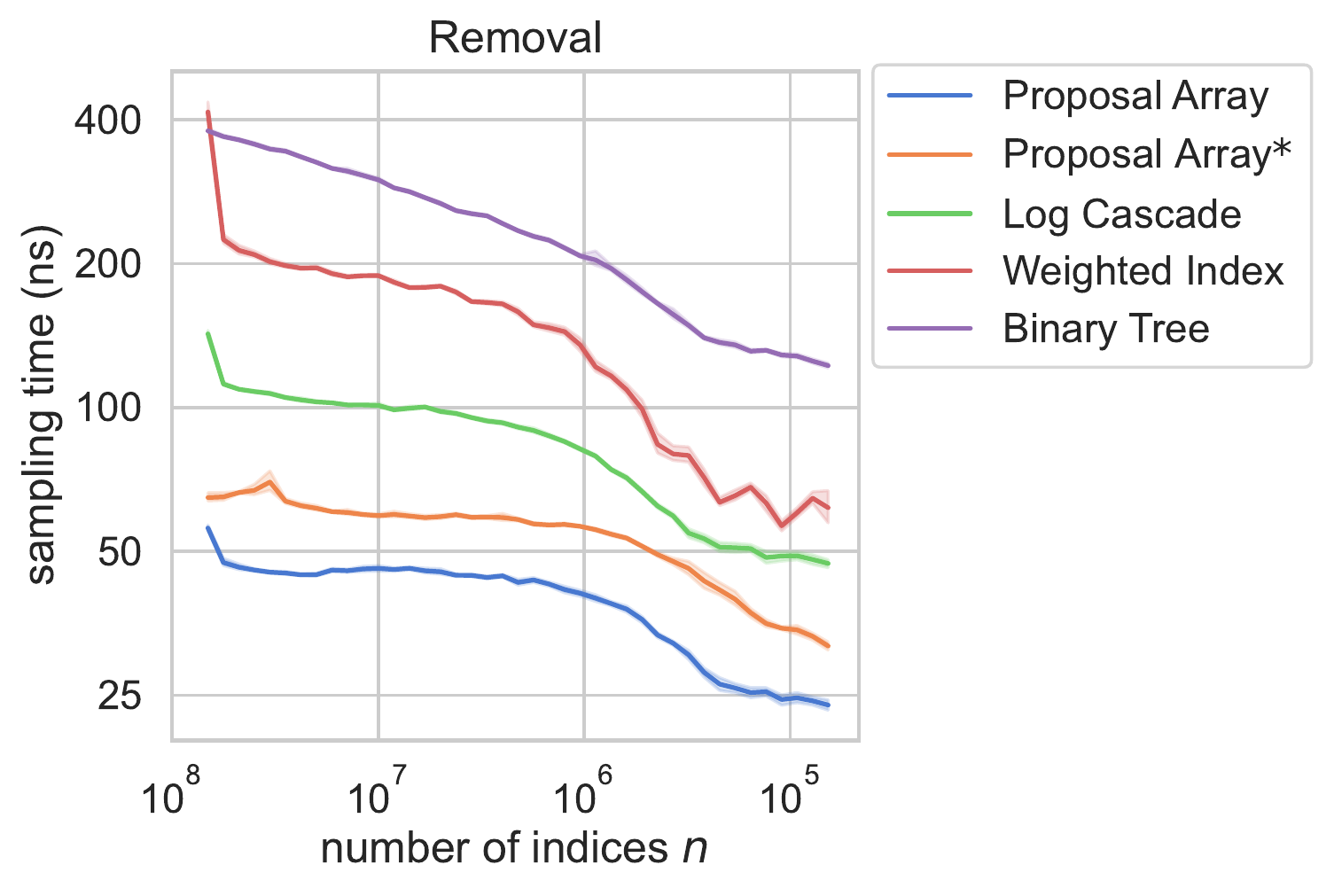}
 \end{subfigure}
 \caption{Average sampling times in nanoseconds of \impl{Proposal Array}, \impl{Proposal Array*}, \impl{Log Cascade}, \impl{Weighted Index} and \impl{Binary Tree} under insertions and removals of indices. The shaded regions indicate the $95\%$ confidence interval.}
  \label{fig:dynamic-sampling-2}
\end{figure}

In the first set of experiments, we compare the sampling performance of all five implementations under insertions and deletions of indices.
To reduce effects not caused by insertions or removals, all weights throughout are drawn uniformly at random from the interval $[0, 10^7) \subset \mathbb R$.
The update patterns are as follows.
\begin{itemize}
	\item \textbf{Insertion}: We start from an index set of size $n = 2^{16}$ and increase the size to $n = 2^{26}$.
	\item \textbf{Removal}: We start from an index set of size $n = 2^{26}$ and decrease the size to $n = 2^{16}$.
\end{itemize}
\autoref{fig:dynamic-sampling-2} plots the average time per sample for both update patterns.
All data structures exhibit a degradation in sampling performance for the \textbf{Insertion} pattern, and an improvement in sampling performance for the \textbf{Removal} pattern due to cache effects.
The relative performances follow the same trend as for changes in weight (see \autoref{fig:dynamic-sampling} below).

\begin{figure}[t!]
 \centering
 \begin{subfigure}{\columnwidth}
  \centering
  \includegraphics[scale=0.55]{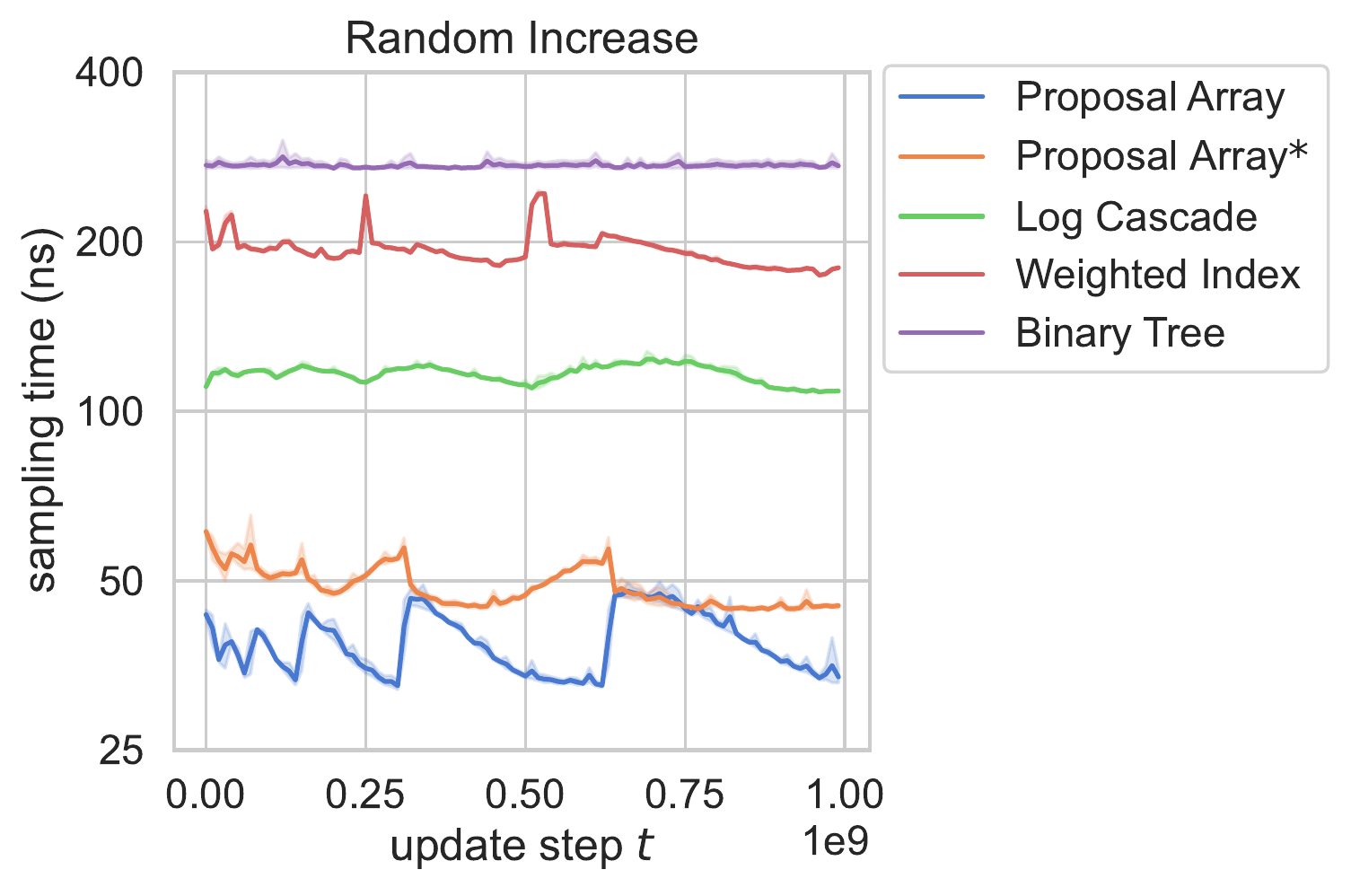}
 \end{subfigure}
 \vfill
 \begin{subfigure}{\columnwidth}
  \centering
  \includegraphics[scale=0.55]{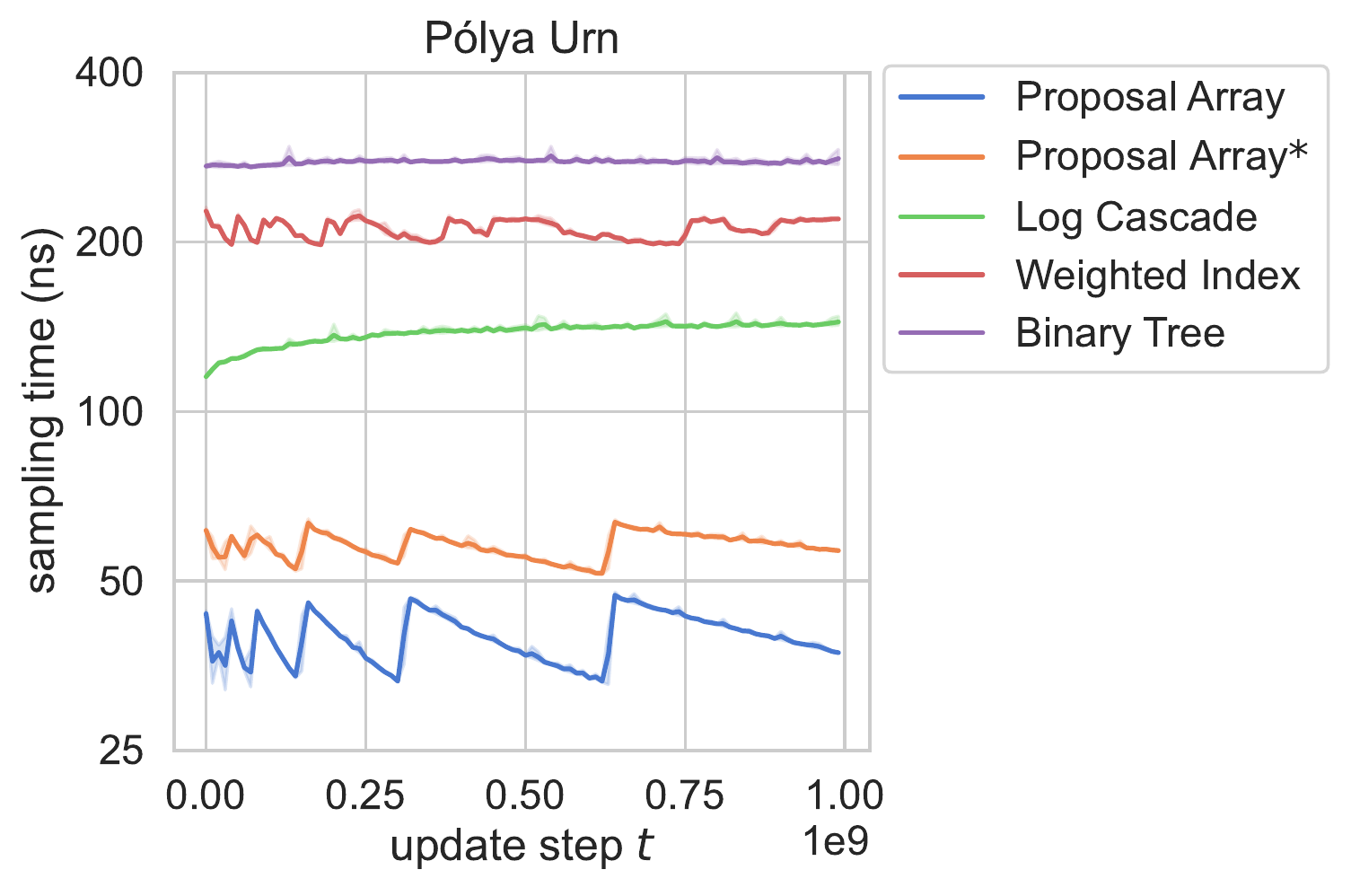}
 \end{subfigure}
 \vfill
 \begin{subfigure}{\columnwidth}
  \centering
  \includegraphics[scale=0.55]{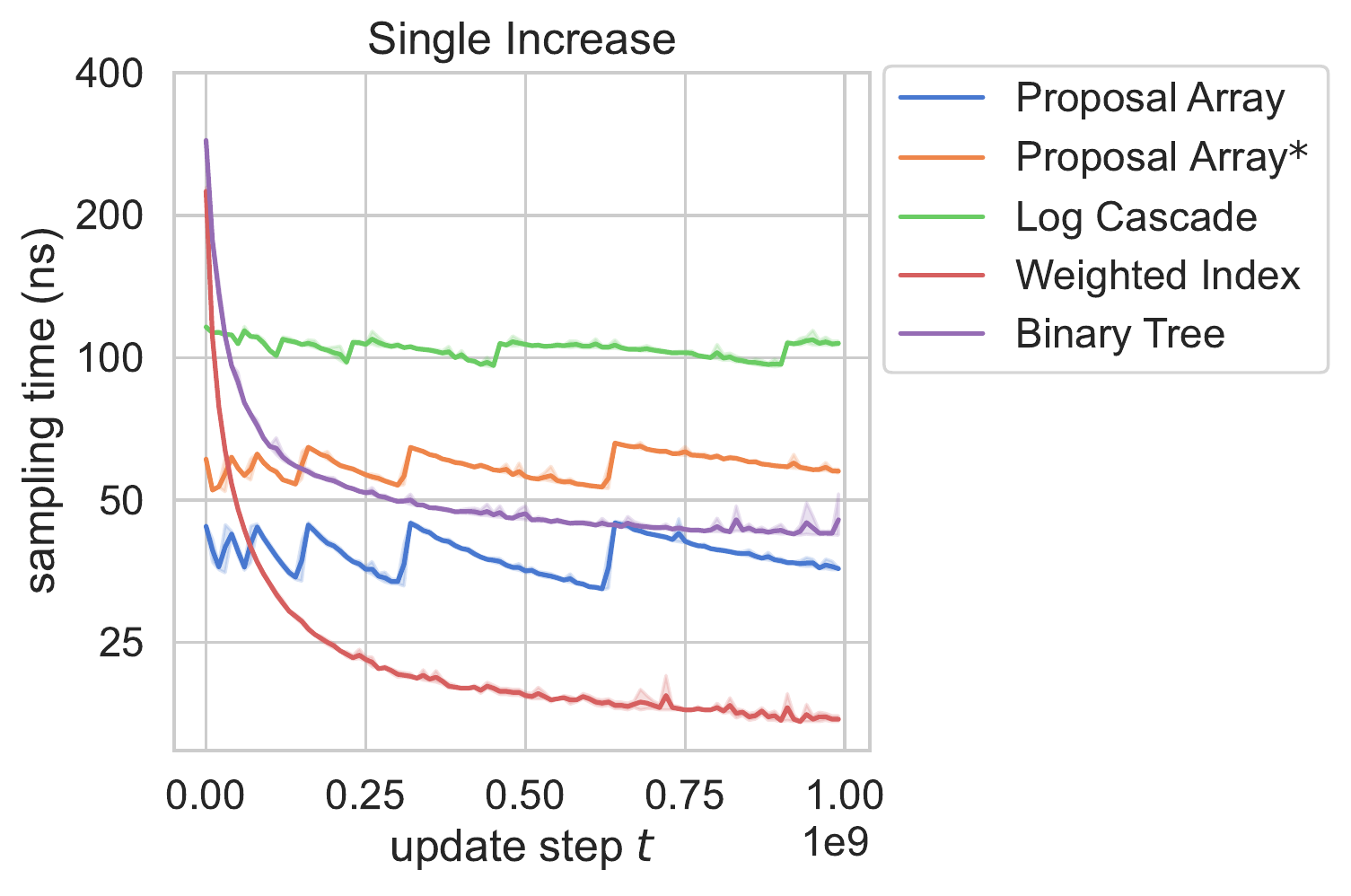}
 \end{subfigure}
 \caption{Average sampling times in nanoseconds of \impl{Proposal Array}, \impl{Proposal Array*}, \impl{Log Cascade}, \impl{Weighted Index}, and \impl{Binary Tree} for various update patterns. The shaded regions indicate the $95\%$ confidence interval.}
  \label{fig:dynamic-sampling}
\end{figure}

In the second set of experiments, we consider changes in weight.
Starting from the \textbf{Noisy} weight distribution (see \autoref{subsec:experiments-static}) on an index set of size $n = 10^7$, we perform $T = 100 n$ update steps.
Once every $T / 100$ update steps, we measure the time required to draw $10^6$ samples and plot the average time per sample as a function of the update step.
We consider the following update patterns where each increase in weight $\Delta_t$ is drawn independently and uniformly at random from the interval $[0, n) \subset \mathbb R$.
\begin{itemize}
	\item \textbf{Random Increase}: In step $t$ we pick an index $i \in S$ uniformly at random and increase its weight by $\Delta_t$.
	\item \textbf{P\'{o}lya Urn}: In step $t$ we sample an index $i \in S$ from the distribution $(S, \mathbf{w})^{(t-1)}$ and increase its weight by $\Delta_t$.
	\item \textbf{Single Increase}: In step $t$ we increase the weight of the index $i = 1$ by $\Delta_t$.
\end{itemize}
Note that the resulting weight distributions in the limit are similar to the weight distributions used in the static sampling experiment (see \autoref{subsec:experiments-static}).

\autoref{fig:dynamic-sampling} shows the results.
We find that the sampling performances of both variants of the dynamic \impl{Proposal Array} are comparable to the static version (compare \autoref{fig:static-sampling}) with \impl{Proposal Array*} being slower by a factor of $1.3$.
Both variants exhibit a regular hacksaw pattern in the sampling time due to the array improving before the mean doubles.

The implementation \impl{Log Cascade} is slower than the \impl{Proposal Array} by a factor of at least $3.1$ for the \textbf{Random Increase} and  \textbf{P\'{o}lya Urn} update patterns but only slower by a factor of $2.7$ for the \textbf{Single Increase} pattern.
The \impl{Weighted Index} is slower than the \impl{Proposal Array} for the \textbf{Random Increase} (factor $5.1$) and \textbf{P\'{o}lya Urn} (factor $5.4$) update patterns but faster by a factor of $1.5$ for the \textbf{Single Increase} patterns.
The dynamic \impl{Binary Tree} behaves similarly to the tree-like \impl{Weighted Index}, but is slower by a factor of $1.4-2.2$.

The speed-up of \impl{Log Cascade} for the \textbf{Single Increase} pattern occurs as the concentration of probability mass on a single index allows most samples to be drawn from the same partition which can then always be held in cache.
For the \impl{Weighted Index} the effect is even more pronounced as the implementation contains optimizations which trivialize the task of sampling from a probability distribution dominated by a single index.
In contrast, the \impl{Proposal Array} and \impl{Proposal Array*} always access a random location in the array and do not benefit even if most locations are occupied by the same index.
Still, this should not present a significant issue in practice, since the optimizations required to improve the performance are straightforward, e.g. if we expect a small subset of indices to dominate the distribution, then these can be treated separately.

\section{Summary}
\label{sec:summary}
We suggest the Proposal Array as a simple and practical method for maintaining a discrete probability distribution under updates of moderate size.
Our experimental results demonstrate that sampling from a dynamic Proposal Array is not slower than sampling from a static Alias Table and faster than sampling from more general but also more complex dynamical solutions.
The variant Proposal Array* improves the update procedure by avoiding reconstructions of the array with a minor slowdown in sampling speed.
Among the fully dynamic solutions, we find that the single-layered Log Cascade performs particularly well.

\bibliography{asa-bibliography}

\appendix

\section{Log Cascade Implementation}
\label{sec:a-log-cascade}
In this appendix, we briefly discuss our implementation of \cite{hagerup1993maintaining}.

The original method is based on the observation that rejection sampling is efficient if the weights $w_1, \dots, w_n$ all fall in a similar range of values.
In particular, the expected number of attempts of rejection sampling is at most $c$ if $w_1, \dots, w_n \geq w_{\text{max}} / c$ for some $w_{\text{max}} \geq w_1, \dots, w_n$.
Moreover, for general weights, the set of indices can be partitioned into subsets $S_{\lceil \log w_{\text{min}} \rceil}, \dots, S_{\lceil \log w_{\text{max}} \rceil}$ where the subset $S_k$ contains all indices whose weights fall into the range $[2^{k-1}, 2^{k})$. 
As rejection sampling can now be used to efficiently sample an index from a subset, it only remains to choose a subset, and it is easy to verify that there are at most $\Oh{\log n}$ subsets to choose from if the weights are polynomial in $n$.
Now iteratively using the partitioning scheme on the subsets, we can further reduce the number of elements to choose from, and in general, after using $k$ layers of partitioning, the remaining number of elements in the top layer is $O(\log^{(k)} n)$ where $\log^{(k)}$ is the iterated logarithm.
Finally, the $\Oh{1}$ sampling time is obtained by choosing $k$ as a large enough constant so that the results of all possible operations on the $N$ remaining elements in the top layer can be pre-computed and stored in a look-up table of size $N^{\Oh{N}}$.

Our implementation uses this same partitioning scheme and supports an arbitrary number of layers.
Similarly to our other implementations, we reduce the number of memory accesses to one per rejection sampling attempt by storing the acceptance probability of an index directly together with the index.
Interestingly, our experiments suggest that using a single partitioning layer outperforms using any larger number of layers (see \autoref{fig:log-cascade-layers}) for moderately large index sets ($n = 10^6$ to $n = 10^8$).
This occurs despite the fact that we use linear sampling instead of the table look-up to sample a subset in the top layer (note that we cannot use the look-up table after one layer of partitioning as the remaining number of elements is still too large).
On the other hand, it seems plausible that linear sampling from between $\log 10^6 \approx 20$ to $\log 10^8 \approx 26$ elements is faster than sampling from two or more partitioning layers, as the latter approach at least doubles the number of calls to the RNG and memory accesses to the acceptance probability of subsets/indices.

\end{document}